\documentclass[a4paper,fleqn]{cas-dc}
\usepackage{graphicx}
\usepackage{subfigure}
\usepackage{colortbl}
\usepackage{booktabs}   
\usepackage{threeparttable}
\usepackage{amsfonts}
\usepackage[numbers]{natbib}
\PassOptionsToPackage{table}{xcolor}
\newtheorem{theorem}{Theorem}
\newtheorem{definition}{Definition}
\newtheorem{protocol}{Protocol}
\newtheorem{proof}{Proof}
\usepackage{amsmath}
\usepackage{amssymb}
\usepackage{multirow}
\usepackage{hyperref}
\usepackage[justification=centering]{caption} 

\usepackage[ruled,vlined,lined,ruled,linesnumbered]{algorithm2e}

\SetKwInput{KwInput}{Input}                % Set the Input
\SetKwInput{KwOutput}{Output}              % set the Output

\def\tsc#1{\csdef{#1}{\textsc{\lowercase{#1}}\xspace}}
\tsc{WGM}
\tsc{QE}
\begin{document}
\let\WriteBookmarks\relax
\def\floatpagepagefraction{1}
\def\textpagefraction{.001}

\shorttitle{HE-DKSAP: Privacy-Preserving Stealth Address Protocol via Additively Homomorphic Encryption}    

\shortauthors{Yuping Yan et al.}  

\title [mode = title]{HE-DKSAP: Privacy-Preserving Stealth Address Protocol via Additively Homomorphic Encryption}  

% \tnotemark[<tnote number>] 


\author[1,3]{Yuping Yan}[orcid=0000-0001-9966-8286]
\fnmark[1]

% \author[<aff no>]{<author name>}[<options>]

%\cormark[<corr mark no>]


\ead{<yupingyan@inf.elte.hu>}

%\ead[url]{<URL>}

\credit{Conceptualization, Methodology, Validation, Writing - Original Draft}

\affiliation[1]{organization={Trustworthy and General AI Lab, School of Engineering, Westlake University},
         %   addressline={}, 
            city={Hangzhou},
%          citysep={}, % Uncomment if no comma needed between city and postcode
            postcode={310024}, 
         %   state={},
            country={PR China}}
\affiliation[2]{organization={Mind Network},
        %    addressline={}, 
           % city={Singapore},
%          citysep={}, % Uncomment if no comma needed between city and postcode
          %  postcode={1117}, 
          %  state={},
            country={Singapore}}
\affiliation[3]{organization={Faculty of Informatics, Department of Computeralgebra},
        %    addressline={}, 
            city={Budapest},
%          citysep={}, % Uncomment if no comma needed between city and postcode
            postcode={1117}, 
          %  state={},
            country={Hungary}}
\affiliation[4]{organization={Ethereum foundation fellow},
        %    addressline={}, 
            city={Zug},
%          citysep={}, % Uncomment if no comma needed between city and postcode
         %   postcode={GU2 7XH}, 
          %  state={},
            country={Switzerland}
         }                         
\affiliation[5]{organization={Department of Computer Science, University of Surrey},
        %    addressline={}, 
            city={Guildford, Surry},
%          citysep={}, % Uncomment if no comma needed between city and postcode
            postcode={GU2 7XH}, 
          %  state={},
            country={UK}}
         
\author[2]{George Shao}
\fnmark[1]
\credit{Conceptualization, Methodology, Validation - Original Draft}
% Footnote of the second author

% Email id of the second author
\ead{george@mindnetwork.xyz}
\author[2]{Dennis Song}
\fnmark[1]
\credit{Conceptualization, Methodology, Validation - Original Draft}
\ead{dennis@mindnetwork.xyz}
\author[4]{Mason Song}
\fnmark[1]
\ead{mason@mindnetwork.xyz}


\credit{Conceptualization, Methodology, Validation - Original Draft}

\author[1,5]{Yaochu Jin}[orcid=0000-0003-1100-0631]

\credit{Supervision, Writing - review \& editing}

\cormark[1]

% Email id of the second author
\ead{jinyaochu@westlake.edu.cn}

\cortext[1]{Corresponding author}



\begin{abstract}
Blockchain transactions have gained widespread adoption across various industries, largely attributable to their unparalleled transparency and robust security features. Nevertheless, this technique introduces various privacy concerns, including pseudonymity, Sybil attacks, and potential susceptibilities to quantum computing, to name a few. In response to these challenges, innovative privacy-enhancing solutions like zero-knowledge proofs, homomorphic encryption, and stealth addresses (SA) have been developed. Among the various schemes, SA stands out as it prevents the association of a blockchain transaction's output with the recipient's public address, thereby ensuring transactional anonymity. However, the basic SA schemes have exhibited vulnerabilities to key leakage and quantum computing attacks. To address these shortcomings, we present a pioneering solution — Homomorphic Encryption-based Dual-Key Stealth Address Protocol (HE-DKSAP), which can be further extended to Fully HE-DKSAP (FHE-DKSAP). By leveraging the power of homomorphic encryption, HE-DKSAP introduces a novel approach to safeguarding transaction privacy and preventing potential quantum computing attacks. This paper delves into the core principles of HE-DKSAP, highlighting its capacity to enhance privacy, scalability, and security in programmable blockchains. Through a comprehensive exploration of its design architecture, security analysis, and practical implementations, this work establishes a privacy-preserving, practical, and efficient stealth address protocol via additively homomorphic encryption.
\end{abstract}


\begin{highlights}
    \item We propose the integration of a homomorphic encryption scheme with a dual key stealth address protocol, aiming to elevate privacy and security levels in blockchain systems. This novel approach is named HE-DKSAP. Furthermore, we are extending this concept to incorporate a fully homomorphic encryption scheme, which we refer to as FHE-DKSAP. 
    
    \item By leveraging the power of homomorphic encryption, HE-DKSAP introduces a novel approach to safeguarding transaction privacy and preventing potential quantum computing attacks. This paper delves into the core principles of HE-DKSAP, highlighting its capacity to enhance privacy, scalability, and security in programmable blockchains.
    \item Our security analysis confirms that HE-DKSAP upholds essential standards such as data confidentiality, transaction unlikability, resistance to quantum computing attacks, and robustness against cipher-policy attacks. Furthermore, our experimental validation demonstrates that HE-DKSAP and FHE-DKSAP both excel in computation efficiency and data storage management.
\end{highlights}

% Keywords

\maketitle

% Main text

\section{Introduction}

Due to transparent transactions, decentralization, and robust security, blockchain has gained significant popularity and adoption in various fields. This transparency ensures that every transaction is openly recorded on a public ledger, accessible to all participants, fostering an environment of trust. On the security front, blockchain employs advanced cryptographic techniques, making it exceedingly difficult for malicious actors to alter transaction data. As a result, businesses and individuals increasingly turn to blockchain for a more transparent, secure, and efficient way to conduct transactions. 

As it is known to all, in traditional blockchain transactions, each participant typically has a digital wallet. A \textit{wallet} is a software or hardware application that stores the user's cryptocurrency holdings and provides the necessary tools for managing and initiating transactions. Wallet addresses are used to send and receive cryptocurrencies, and these are alphanumeric strings that serve as unique identifiers for users on the blockchain. These addresses are derived from the user's public key using cryptographic techniques. Behind each wallet address, there are corresponding public and private keys. A public key is used to generate the wallet address and is publicly known. By contrast, a private key is kept secret and is used to sign transactions, proving ownership of the funds associated with the wallet address. However, contrary to the common perception that blockchain transactions provide complete anonymity, they are, by default, pseudonymous. This pseudonymity implies that blockchain transactions openly reveal the sender and receiver addresses, as well as the transaction amount, all of which are recorded on the public ledger. 

Ensuring privacy within the blockchain ecosystem presents a substantial challenge. Various privacy-preserving techniques have emerged to address the challenge of pseudonymity in traditional blockchain transactions, such as Zero-Knowledge Proofs \cite{blum2019non}, encryption schemes, and Stealth Addresses (SA), to name a few. Among them, introducing the concept of SA as a centric solution for privacy-preserving transactions is noteworthy. Based on the Diffie-Hellman key exchange protocol \cite{diffie2022new}, SA can protect users' privacy by making it extremely difficult for outsiders to link a transaction to a specific individual or entity. This anonymity is essential as blockchain ledgers are public, and anyone can inspect them. They achieve this by generating unique, one-time addresses for each transaction, making it challenging for external parties to track or analyze financial activity on the blockchain. 

Nowadays, one of the most widely embraced SA schemes is the Dual-Key Stealth Address Protocols (DKSAP) \cite{courtois2017stealth}, which introduces a noteworthy enhancement by incorporating multiple distinct spending keys into the protocol. This strategic addition serves to fortify the protocol's resilience against potential attacks like the “bad random attack" or situations where keys become compromised. It separates the root keys into viewing and spending keys to reduce the risk of the root key. While this advancement offers enhanced security, it does come with the risk of viewing key leakage and vulnerability to quantum computing attacks. We conclude that the current SA schemes face three primary challenges: key leakage attacks, scalability and usability concerns, and vulnerability to quantum computing attacks.

%Recently, Toni et al. introduced an innovative SA design, known as EIP-5564, marking a significant stride in the field. They further materialized this design through the creation of BasedSAP \cite{wahrstatter2023basesap}, an Ethereum-based implementation of the SA concept. BasedSAP provided a reliable way to implement SA on the application layer of programmable blockchains like Ethereum, utilizing the Secp256k1 elliptic curve.

\begin{itemize}
    \item \textbf{Temporary key leakage \cite{feng2020pdksap}:} One significant challenge in SA Protocols is the vulnerability to key leakage attacks. These attacks occur when every time a transaction is made, the temporary public key of the recipient is attached. This makes stealth transactions easily identifiable. Specifically, the presence of the public key indicates that a particular transaction is a stealth transaction. 
    \item \textbf{Scalability and Usability:} As the adoption of blockchain networks like Ethereum continues to grow, scalability and usability become crucial concerns for any protocol, including SA protocols. Generating unique stealth addresses and managing multiple spending keys can create usability challenges for users, especially when users make frequent transactions. Striking a balance between robust privacy measures and user-friendly experiences is essential to encourage widespread adoption and make the protocol accessible to non-technical users.
    \item \textbf{Quantum Computing Threats:} The advent of quantum computing \cite{steane1998quantum} presents a potential threat to the security of existing cryptographic systems. Quantum computers have the potential to solve complex mathematical problems that underpin many encryption methods efficiently. This could render the privacy-enhancing features of SA protocols vulnerable to attacks. Researchers are exploring cryptographic techniques resistant to quantum computing threats to ensure the longevity of these privacy solutions.
\end{itemize}

%To overcome these challenges, ongoing research and development efforts are crucial. Solutions might involve incorporating advanced encryption techniques, exploring new cryptographic primitives, and refining the protocol's design to provide a more resilient and user-friendly experience while maintaining strong privacy guarantees. Based on BaseSAP, we contribute further to propose FHE-DKSAP: a SA protocol with Fully Homomorphic Encryption (FHE). The main contribution of this paper includes the follows: 
To overcome these challenges, we propose the HE-DKSAP (Homomorphic Encryption-based Dual-Key Stealth Address Protocol) to provide a more resilient and secure framework while maintaining strong privacy guarantees. Fully homomorphic encryption, based on lattice, allows users to conduct computation over ciphertext without revealing the raw information. HE-DKSAP can be extended to Fully HE-DKSAP (FHE-DKSAP) for quantum computing resistance. The main contributions of this work include the following: 

\begin{itemize}
    \item HE-DKSAP replaces elliptic curve \cite{blake1999elliptic} with homomorphic encryption to improve security level by hiding the raw information. Fully homomorphic encryption constructs the lattice cryptographic and is born to equip fully HE-DKSAP to prevent quantum computing attacks. Therefore, SA in FHE-DKSAP is quantum-safe. 

    \item HE-DKSAP allows for the reuse of the recipient's public key, which not only reduces the complexities of key management but also eliminates the necessity for an ever-growing number of keys. It reduces the complexity and difficulty of SA adoption without generating a large of stealth addresses.

    \item Contrary to the dual-key framework utilized in DKSAP, our novel HE-DKSAP design addresses the issue of temporary key exposure by replacing the viewing key with the FHE key pairs, leveraging the benefits of HE schemes. 

\end{itemize}

The rest of the paper is organized as follows. Section II presents the related work of this work. Section III provides the preliminaries essential for understanding the proposal. The main proposed Homomorphic Encryption-based dual key stealth Address protocol (HE-DKSAP) approach is detailed in Section IV, followed by the implementation in Section V and evaluation in Section VI with security analysis and experiment validation. Finally, Section VII concludes the paper and proposes future work. 

\section{Related work}
This section reviews the body of work related to our proposed approach, focusing on the integration of stealth addresses and the fully homomorphic encryption scheme.

\subsection{Stealth addresses}
The development of the Stealth Addresses (SA) technology began with its initial invention by a user named ‘bytecoin’ in the Bitcoin forum on April 17, 2011. This technique introduced the concept of untraceable transactions capable of carrying secure messages, paving the way for enhanced privacy and security in blockchain systems. In 2013, van Saberhagen took the concept further in the CryptoNote white paper \cite{van2013cryptonote}, providing more insights and advancements in the SA technology. His contribution expanded the understanding of how SA can be integrated into cryptographic protocols. 

In 2017, Courtois and Mercer introduced the Robust Multi-Key SA \cite{courtois2017stealth}, which first proposed the dual key SA protocol (DKSAP) and enhanced the robustness and security of the SA technique. Fan et al. presented a faster dual-key SA protocol \cite{fan2018faster} based on the DKSAP and designed explicitly for blockchain-based Internet of Things (IoT) systems. Their protocol introduced an increasing counter, enabling quicker parsing and improving efficiency. Fan et al. tackled the issue of key length in Stealth Addresses by utilizing bilinear maps \cite{fan2019new}, thereby significantly enhancing the protocol's security and practicality. Meanwhile, Liu et al. introduced a lattice-based linkable ring signature supporting SA \cite{liu2019lattice}. This innovation aimed at countering adversarially chosen-key attacks, further reinforcing the security aspect. A ring signature and SA-based privacy-preserving blockchain-based framework for health information exchange was also proposed in \cite{lee2021mexchange}. Most recently, a fully open and reusable SA protocol, BaseSAP \cite{wahrstatter2023basesap}, was reported to provide a reliable way to implement SA on the application layer of programmable blockchains like Ethereum, utilizing the Secp256k1 elliptic curve and integrating other cryptographic techniques like elliptic curve pairings and lattice-based cryptography. 

\subsection{Fully homomorphic encryption}
Homomorphic encryption (HE) has been called “The Swiss Army knife of cryptography'', which allows one to perform complex mathematical data operations directly on encrypted data without compromising the encryption. HE is a powerful tool for safeguarding data privacy while simultaneously performing computational tasks. Based on the (Ring) Learning with Error (LWE) assumption \cite{lyubashevsky2010ideal}, HE on lattice-based cryptography demonstrates effectiveness in thwarting post-quantum attacks. The basic operations of a typical HE scheme encompass addition, multiplication, or both. These capabilities form the basis for the categorization of HE schemes, which is listed below:

  \begin{itemize}

        \item Partially Homomorphic Encryption (PHE). The initial stage of homomorphic encryption is called partial homomorphic encryption. The definition is that the ciphertext has only one homomorphic characteristic. This stage includes either the mode of additive homomorphism or multiplicative homomorphism, such as Paillier's cryptosystem \cite{Paillier99}, which can only support plaintext addition.

        \item Leveled/ Somewhat Homomorphic Encryption (LHE/ SWHE) \cite{brakerski2014leveled}. This allows any combination of addition and multiplication operations to be performed on the ciphertext. However, there is a limit to the number of such operations that can be performed, which is why it is called leveled or somewhat homomorphic encryption.

        \item Fully Homomorphic Encryption (FHE) \cite{Gentry09}. This permits a combination of addition and multiplication operations on the ciphertext. Taking advantage of Bootstrapping, a fully homomorphic encrypted system is achieved without any computational limitations.

   \end{itemize}

\section{Preliminaries}
In this section, we will present the preliminaries related to the HE-DKSAP, including the Diffile-Helleman key exchange protocol, DKSAP, Paillier scheme of homomorphic encryption, and fully homomorphic encryption scheme.

\subsection{Diffile-Hellman key exchange protocol}
    In Diffie-Hellman, they use the fact that the multiplicative group $Z^*_p$ for a prime $p$ is \emph{cyclic}. This means that there is an element $g \in Z^*_p$ called \emph{generator}, such that $Z^*_p = \{ 1 ,g ,g^2,g^3,\ldots, g^{p-2} \}$. The Diffie-Hellman (DH) key exchange protocol requires a group $(G,\cdot)$ and a generator of the group $g\in G$. The details of the DH protocol can be found as follows: 

\begin{protocol}[Diffile-Hellman key exchange protocol]
This protocol defines how a secrete key can be exchanged without sharing sensitive information.
\label{he}
\begin{itemize}
    \item[a.] There are global common elements: random prime number $p$ and a generator $g$.
    \item[b.] Alice chooses a secret $a$, and calculates public key $Y_A = g^a \pmod p$ and sends it to Bob.
    \item[c.] Bob chooses a secret $b$, and calculates public key $Y_B = g^b \pmod p$ and send it to Alice.
    \item[d.] Alice and Bob both compute $sk=g^{ab} \pmod{p}$ by computing $Y_B^a$ and $Y_A^a$, respectively.
    \item[e.] $sk$ is the private key for encryption for both parties.
\end{itemize}
\end{protocol}

\subsection{Dual-Key Stealth Address Protocol (DKSAP)}

DKSAP builds on the Diffile-Hellman key exchange protocol in the elliptic curve. $G$ is the generator in the elliptic curve. When sender A wishes to transmit a transaction to receiver B in stealth mode, DKSAP operates as illustrated in Fig. \ref{dksap}, and the process is summarized as follows: 
\begin{definition}[Stealth meta-address]
    A “stealth meta-address” is a set of one or two public keys that can be used to compute a stealth address for a given recipient.
\end{definition}
\begin{definition}[Spending key]
    A “spending key” is a private key that can be used to spend funds sent to a stealth address. 
\end{definition}
\begin{definition}[Viewing key]
    A “viewing key” is a private key that can be used to determine if funds sent to a stealth address belong to the recipient who controls the corresponding spending key. 
\end{definition}
\begin{protocol}[DKSAP]
The details of the DKSAP are the follows: 
\begin{enumerate}
    \item The receiver B has a pair of private/public keys $(v_b, V_b)$ and $(s_b, S_b)$, where $v_b$ and $s_b$ are called B’s “scan private key" and “spend private key", respectively, whereas $V_b = v_b*G$ and $S_b = s_b*G$ are the corresponding public keys. Note that none of $V_b$ and $S_b$ ever appear in the blockchain, and only the sender A and the receiver B know those keys. 
    \item The sender A generates an ephemeral key pair $(r_a, R_a)$ with $R_a = r_a*G$ and $0 < r_a < n$, and sends $R_a$ to the receiver B. 
    \item Both the sender A and the receiver B can perform the elliptic curve Diffile-Helleman protocol to compute a shared secret: $c_{AB} = H(r_{a}*v_b*G) = H(r_{a}*V_b) = H(v_{a}*R_a)$, where $H()$ is a cryptographic hash function.
    \item The sender A can now generate the destination address of the receiver B to which A should send the payment: $T_a = c_{AB}*G + S_b$. Note that the one-time destination address $T_a$ is publicly visible and appears on the blockchain. 
    \item Depending on whether the wallet is encrypted, the receiver B can compute the same destination address in two different ways: $T_a = c_{AB}*G + S_b = (c_{AB} + s_b)*G$. The corresponding ephemeral private key is $t_a = c_{AB} + s_b$, which can only be computed by the receiver B, thereby enabling B to spend the payment received from A later.
\end{enumerate}

\end{protocol}

\begin{figure*}[htb]
\centering
  \includegraphics[width=0.9\linewidth]{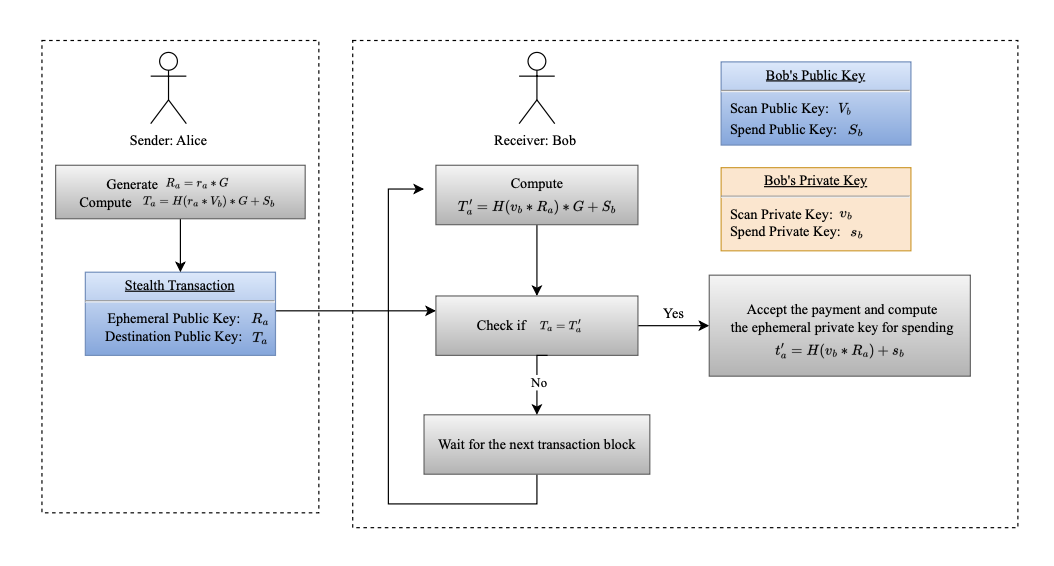}\\
  \caption{Dual-key Stealth Address Protocol (DKSAP)}
  \label{dksap}
\end{figure*} 

\subsection{Homomorphic encryption - Paillier}

Homomorphic encryption is a powerful tool for safeguarding data privacy while simultaneously performing computational tasks. The definition of HE can be found as follows:
\begin{definition}[Homomorphic encryption]
    Given two sets A and B, and a function $f: A \rightarrow B$, $f$ is said to be an additive homomorphism if, for any $x$ and $y$ in $A$, we have: $f(x \diamond y) = f(x) \diamond f(y)$.
    If the operation $\diamond$ represents an addition, we define it as an additive homomorphism. On the other hand, if $\diamond$ represents multiplication, we define it as a multiplicative homomorphism.
\end{definition}

The initial stage of homomorphic encryption is called partial homomorphic encryption. The definition is that the ciphertext has only one homomorphic characteristic. This stage includes either the mode of additive homomorphism or multiplicative homomorphism, such as Paillier's cryptosystem \cite{paillier1999public}, which can only support plaintext addition. The Paillier scheme stands out as perhaps the most well-known and the most efficient in its category of partially homomorphic systems. It uniquely supports the addition of ciphertexts, allowing for operations on encrypted data without needing to decrypt it.  

The Paillier scheme consists of key generation, encryption, and decryption. $p$ and $q$ are two large prime numbers, $n$ is the product of these two primes, $\lambda$ is Carmichael's function and $L$ is Lagrange function. The algorithm can be found as follows:

\begin{itemize}
    \item \textbf{Key generation:} 
    \begin{enumerate}
        \item Choose two large prime numbers $p$ and $q$ randomly and independently of each other such that $gcd(pq,(p-1)(q-1)) = 1$. This ensures $p$ and $q$ are co-prime.
        \item Compute the modulus $n$ - the product of two primes $n=pq$ and $\lambda = lcm(p-1,q-1)$, where $lcm$ denotes the least common multiple.
        \item Select random integer $g$ where $g \in \mathbb{Z}_n^{ 2}$.
        \item Ensure $n$ divides the order of $\mathrm{g}$ by checking the existence of the following modular multiplicative inverse: $\mu=L\left(g ^\lambda \operatorname{modn} n^2\right)^{-1} \bmod n$,  where $L$ is Lagrange function and $L(\mu)=\mu-1 / n$ for $\mu=1 \bmod n$.
        \item The public (encryption) key is $(n, g)$.
        \item The private (decryption) key is $(\lambda, \mu)$.
    \end{enumerate}

 \item \textbf{Encryption:}  
\begin{enumerate}
    \item Given a message $m$ where $0 \leq m<n$, and the public key $(n, g)$. 
    \item Choose a random $r$ where $r$ belongs to $\mathbb{Z}_n^*$.
    \item Compute the ciphertext $c$ as:
    \begin{equation}
        c=g^mr^n \bmod n^2
    \end{equation}
\end{enumerate}

 \item \textbf{Decryption:}  
 \begin{enumerate}
     \item Given a ciphertext $c$ and the private key $(\lambda, \mu)$.
     \item Compute the plaintext $m$ as:
     \begin{equation}
         m=L\left(c ^\lambda \bmod n^2\right) / L\left(g ^\lambda \bmod n^2\right) \bmod n 
          \end{equation}
Or in the same expression:
\begin{equation}
    m=L\left(c ^\lambda \bmod n^2\right) \cdot \mu \bmod n.
\end{equation}
\end{enumerate}
\end{itemize}

To sum up, the additive homomorphism property is exhibited as:
\begin{equation}
    E\left(m_1\right) \times E\left(m_2\right) \bmod n^2=E\left(m_1+m_2 \bmod n\right).
\end{equation}

\subsection{Fully homomorphic encryption scheme}
Taking advantage of Bootstrapping \cite{gentry2009fully}, a fully homomorphic encrypted system \cite{fan2012somewhat} is achieved by supporting both the addition and multiplication computation operations without any computational limitations. In the context of the Learning with Error (LWE) problem, Fully Homomorphic Encryption (FHE) emerges as an efficient defense against post-quantum attacks. In the following, we will give an overview of the FHE scheme, explain the LWE problem, and give an example of the FHE scheme - BFV.
\begin{definition}(Fully homomorphic encryption scheme)
A fully homomorphic encryption scheme consists of a set of probabilistic polynomial time algorithms defined as follows: $E$ = (KeyGen, Encrypt, Decrypt, Evaluate):
\begin{itemize}
    \item $KeyGen\left(1^\lambda\right) \rightarrow(sk, pk, evk)$: Given the security parameter $\lambda$, outputs a key pair consisting of a public encryption key $pk$, a secret decryption key $sk$, and an evaluation key $evk$.
    \item $Enc(pk, m) \rightarrow ct$: Given a message $m \in \mathcal{M}$ and an encryption key $pk$, outputs a ciphertext $ct$.
    \item $Dec(sk, ct)=m$: Given the secret decryption key and a ciphertext $ct$ encrypting $m$, outputs $m$.
    \item  $Eval\left(evk, f, ct_1, ct_2, \ldots, ct_n\right) \rightarrow ct^{\prime}$: Given the evaluation key, a description of a function $f: \mathcal{M}^n \rightarrow \mathcal{M}$, and $n$ ciphertexts encrypting messages $m_1, \ldots, m_n$, outputs the result ciphertext $ct^{\prime}$ encrypting $m^{\prime}=$ $f\left(m_1, \ldots, m_n\right)$.
\end{itemize}
\end{definition}
\subsubsection{Learning with Errors}
The Learning With Errors Assumption (LWE) is a cryptographic assumption over vectors and learning algorithms. Given a probabilistic polynomial-time (PPT) algorithm, which takes a security parameter $1^n$ and produces a vector $\vec{s} \in \mathbb{Z}_q^n$ along with a distribution $\chi$ and a sample count $m$ in poly $(n)$, the Learning With Errors Assumption comes, if:
    \begin{align}
        (\vec{s}, n, q, m, \chi) \leftarrow \operatorname{LWEGen}\left(1^n\right),\\
        A \leftarrow \mathbb{Z}_q^{n \times m}, \vec{e} \leftarrow \chi^m, \vec{b} \leftarrow \mathbb{Z}_q^m.
    \end{align}

Then
\begin{align}
    (A, \vec{s} A+\vec{e}) \stackrel{c}{\equiv}(A, \vec{b}).
\end{align}

In other words, $m$ random linear combinations of the elements of $\vec{s}$ with errors $\vec{e}$ cannot be efficiently distinguished by $m$ uniform samples (given by $\vec{b}$), even if matrix $A$ that generates the combinations is known. If we select $\vec{s}$ and $\Vec{e}$ to be one-dimensional matrices, $A$ will be a two-dimensional matrix, and $\vec{b}$ will be a one-dimensional vector.

Ring Learning with Errors (RLWE) is more properly called learning with errors over rings. It is the more extensive learning with errors problem specialized to polynomial rings over finite fields. The polynomial ring over a finite field is defined as $\mathbb{Z}_q[x] /\left\langle x^n+1\right\rangle$, where $x^n+1$ is the quotient (like the modulus in mod operation) so that no polynomial has a rank higher than $n$ in the field and $q$ also applies mod operation on the coefficients.

\subsubsection{BFV scheme}
We introduce one of the notable FHE schemes known as the Brakerski-Fan-Vercauteren (BFV) \cite{fan2012somewhat} scheme. It encrypts polynomials instead of bits and supports the addition and multiplication of vectors in a homomorphic way. As an FHE scheme, it includes the setup, key generation, encryption, and decryption process as well. 
\begin{enumerate}
    \item  $\operatorname{Step}(\lambda)$. For a security parameter $\lambda$, we set a ring size $n$, the polynomial ring $R=Z[X] /\left(X^{\wedge} n+1\right)$, a ciphertext modulus $q$, a special modulus $p$ coprime to $q$, a key distribution $\chi$, and an error distribution $\Omega$ over $R$. 
    \item Key generation. We sample $\mathbf{s} \leftarrow \chi$ and output $sk=\mathbf{s}$ while the $pk=\left([-(\mathbf{a} \cdot \mathbf{s}+\mathbf{e})]_q, \mathbf{a}\right)$, $\mathbf{a} \leftarrow R_q, \mathbf{e} \leftarrow \chi$.
    \item Encryption. To encrypt a message $m$, we sample $u, e_0, e_1$ from the error distribution $\chi$. We return the ciphertext as:
    \begin{equation}
        \mathbf{ct}=\left(\mathbf{pk}[0] \cdot u+e_0, \mathbf{pk}[1] \cdot u+e_1+\Delta \cdot m\right),
    \end{equation}
    where $\Delta=\lfloor q / t\rfloor$. To be more specific, we set $\mathbf{pk}[0] \cdot u+e_0$ as $ct_0$, and $\mathbf{pk}[1] \cdot u+e_1+\Delta \cdot m$ as $ct_1$. 
    
    %A BFV ciphertext is encrypted as a vector of two polynomials $\left(ct_0, ct_1\right) \in R_q^2$. Specifically, we have $ct_0=$ $-a$ and $ct_1=a \cdot s+\frac{q}{p} m+e_0$, where $a$ is a uniformly sampled polynomial. $s$ and $e_0$ are polynomials whose coefficients drawn from $\chi_\sigma$ ($\sigma$ is the standard deviation).
    \item  Decryption. The decryption computes:
    \begin{equation}
        m^{\prime} \leftarrow \mathbf{ct}[0] \cdot \mathbf{s}+\mathbf{ct}[1] \bmod q
    \end{equation}

 %   $\frac{p}{q}\left(ct_0 \cdot s+ct_1\right)=$ $m+\frac{p}{q} e_0$. To successfully decrypt the ciphertext, a bootstrapping is required to keep $e_0$ in check. 
\end{enumerate}

\section{Protocol Overview}
Based on the previous background and primitives, we will provide an overview of our protocol. First, we will clarify the security requirements and then describe the details of the Homomorphic encryption-based Dual Key Stealth Address Protocol (HE-DKSAP), together with the detailed algorithm and protocol flow. 

\subsection{Security requirements}
We provide the security definitions for HE-DKSAP, including data confidentiality, unlinkability, quantum computing attack resistance, and Chosen-plaintext attack security. In the framework of this protocol, we assume that both of the users (sender and the receiver) are semi-honest, which means they follow the protocol honestly but try to discover as much information as possible. Meanwhile, we assume that the communication channel is public and insecure, which is more practical in real scenarios.  

\begin{definition}(Data confidentiality)
    Any probabilistic polynomial time adversary party can learn any users' private information, mainly the private keys with a negligible probability only.  
\end{definition}

\begin{definition}(Unlinkablity)
        Any probabilistic polynomial time adversary party can link the generated stealth address with the real receiver's address with a negligible probability only.  
\end{definition}

%\begin{definition}(Correctness)
%    The FHE-DKSAP scheme can correctly calculate the stealth address after fully homomorphic encryption and decryption algorithms. 
%\end{definition}

\begin{definition}(Quantum computing attack resistance)
    The FHE-DKSAP scheme can prevent quantum computing attacks. 
\end{definition}

\begin{definition}(Chosen-plaintext attack (CPA) security \cite{barrera2010chosen})
Considering the following game between an adversary $\mathcal{A}$ and a challenger:\\
1. The challenger chooses a key $k \leftarrow\{0,1\}^n$. \\
2. The adversary repeatedly chooses messages $m_i$ and the challenger sends back ciphertexts $c_i=\operatorname{Enc}\left(k, m_i\right)$. The adversary can do this as many times as it wants. \\
3. The adversary chooses two challenge messages $m_0^*, m_1^*$. The challenger sends back a ciphertext $c^*=\operatorname{Enc}\left(k, m_b^*\right)$.\\
4. The adversary repeatedly chooses messages $m_i$ and the challenger sends back ciphertexts $c_i=\operatorname{Enc}\left(k, m_i\right)$. The adversary can do this as many times as it wants.\\
5. The adversary outputs guess $b^{\prime} \in\{0,1\}$.

An encryption scheme is CPA secure if for all PPT adversaries $\mathcal{A}$:
\begin{equation}
    \operatorname{Adv}_{\mathcal{A}}^{\mathrm{cpa}}(\lambda)=\left|\operatorname{Pr}\left[b^{\prime}=b\right]-\frac{1}{2}\right|
\end{equation}
is negligible in $\lambda$.
\end{definition}

\subsection{HE-DKSAP Scheme}

In this protocol, there are two participants: Alice, who acts as the sender, and Bob, who is the receiver. Both parties aim to generate and agree upon a stealth address within the HE setting, and only the receiver, Bob, can recover this stealth address. This ensures no one can trace or extract sensitive information from the receiver. To reach a higher security level, we propose a stealth address protocol with homomorphic encryption to prevent quantum computing attacks based on the Learning with Error (LWE) security assumption. While the original proposal of SA builds on the dual-key, our approach, HE-DKSAP, can help the receiver outsource the computation of checking the entire chain for stealth addresses containing assets without revealing his temporary key and prevent quantum computing attacks based on the lattice cryptographic construction that relies on far simpler mathematics than elliptic curve isogenies. Compared to the original version of DKSAP and BaseSAP \cite{wahrstatter2023basesap}, our approach provides full privacy computing for SA without the leakages of keys and personal information. We present HE-DKSAP with details below. The diagram can be found in Fig. \ref{fhe-dksap}, and the description can be found in Protocol \ref{he-dksap}.

In Algorithm \ref{HE-dksap alg}, we present the HE-DKSAP algorithm. There are two primary types of key pairs in this context. The first type is the homomorphic key pairs produced by HE algorithms like Paillier and the BFV algorithm, which we previously discussed. The second type pertains to wallet key pairs. Here, the private key is a randomly generated 64-hex-character key. Its corresponding public key is generated using the equation $PK = sk * G$, $G$ is the generator in secp256k1 elliptic curve \cite{mayer2016ecdsa}. The secp256k1 elliptic curve is defined with the equation $y^2 = x^3 + 7 (mod p)$ over the finite field $Z(2)(256)(-2)(32)(-977)$, which means the $X$ and $Y$ coordinates are 256-bit integers modulo a large number. The main functions consist of public key generation based on the secret key. From the public key, we can get the address by the $pktoaddress$ function. The process begins by hashing the public key's value with the keccak-256 \cite{vujivcic2018blockchain} hash function. We then take the last 20 bytes and prepend them with '0X' to form the stealth address. Within this algorithm, the functions $HE_{encrypt}$ and $HE_{decrypt}$ adhere to the encryption and decryption principles of the designated HE algorithm, as detailed in Section III.

\begin{algorithm}
\normalsize
 \SetAlgoLined
 \SetKwData{Left}{left}\SetKwData{This}{this}\SetKwData{Up}{up}
 \SetKwRepeat{doWhile}{do}{while}
 \SetKwFunction{Union}{Union}\SetKwFunction{FindCompress}{FindCompress}
 \SetKwInOut{Input}{Input}\SetKwInOut{Output}{Output}
 \Input{Ephemeral secret key of Alice $sk_1$, spending secret key of Bob $sk_2$, homomorphic encryption secret key $sk_b$, homomorphic encryption public key $PK_b$, $G$ is the generator of Secp256k1 elliptic curve}
 \Output{Stealth addresses $SA$}
 Key generation function: $PK_1 =sk_1 * G , PK_2 = sk_2 *G$\\
 \textbf{Bob side:}\\
Encrypt the $sk_2$: $C_2 = HE_{encrypt}(sk_2,PK_b)$\\
Send $C_2$ to Alice;\\
  \textbf{Alice side:} \\
 Encrypt the $sk_1$: $C_1 = HE_{encrypt}(sk_1,PK_b)$, ($HE_{encrypt}$ is the encryption function of the HE scheme)\\
 Conduct the additive homomorphic encryption: $C = C_1 + C_2$\\
  $PK_z = PK_1 + PK_2$ \\
$SA = pktoaddress(PK_z) =$
$ '0x' + last 20 bytes(keccak-256_{hashed}(PK_z))$\\
Send $C$ to Bob;\\
\textbf{Bob side:}\\
%Encrypt the $sk_2$: $C_2 = HE_{encrypt}(sk_2,PK_b)$\\
%Conduct the additive homomorhic encryption: $C = C_1 + C_2$\\
Decrypt the ciphertext to get secret key: $sk_z = HE_{decrypt}(C,sk_b)$,
($HE_{decrypt}$ is the decryption function of the HE scheme)\\
$PK_z =sk_z * G$ \\
Compute the stealth addresses $SA = pktoaddress(PK_z) =$
$'0x' + last 20 bytes of (keccak-256_{hashed}(PK_z))$

\caption{HE-DKSAP}
\label{HE-dksap alg}    
\end{algorithm}

\begin{figure*}[htb]
\centering
  \includegraphics[width=0.7\linewidth]{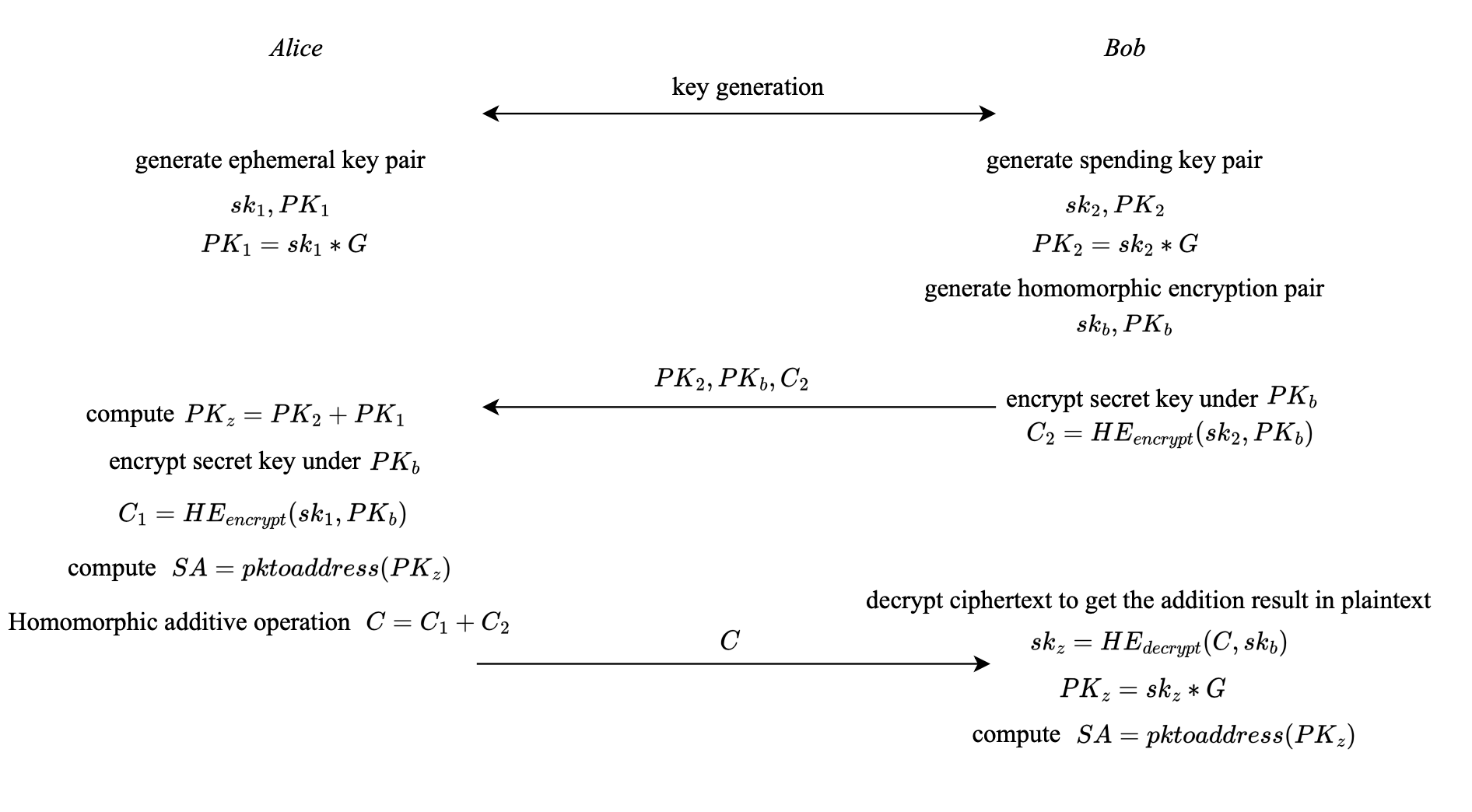}\\
  \caption{Homomorphic Encryption based Dual-key Stealth Address Protocol (HE-DKSAP)}
  \label{fhe-dksap}
\end{figure*} 

\begin{protocol}(HE-DKSAP)
\label{he-dksap}
\begin{enumerate}
    \item 
    \begin{itemize}
        \item Bob (receiver) creates two key pairs: $\left(sk_2, PK_2\right)$ and $\left(sk_b, PK_b\right)$.
        \item $sk_2$ is a randomly generated private key for $\mathrm{SA}$ spending purpose.
        \item  An SA spending wallet address public key $P K_2$ is generated using $sk_2$. It follows the standard Ethereum address conversion from $sk_2$ to $PK_2$. As said, the final wallet address by $PK_2$ does not need to be registered on Ethereum before use.
        \item $sk_b$ is the HE private key for SA encryption and decryption. It is generated by the HE key generation functions.
        \item Bob encrypts the value of $sk_2$ under $PK_b$ to get the ciphertext $C_2$.
        \item Bob publicly shares $PK_2, PK_b, C_2$.
    \end{itemize}
    \item 
    \begin{itemize}
        \item Alice (sender) generates a key pair $\left(sk_1, PK_1\right)$ randomly for each SA transaction.
        \item $sk_1$ is Ethereum ephemeral, and the public key or wallet address does not need to be registered on Ethereum before use.
        \item She combines the two public keys for Ethereum wallet generation, $PK_1$, and $PK_b$, to obtain $PK_z$.
        \item The Stealth Address (SA) is generated based on $P_z$ by following standard Ethereum address conversion. Here, we define this as the public key to address function, and the details can be found in the Algorithm \ref{HE-dksap alg}.
        \item Alice encrypts the secret key $sk_1$ using Bob's HE public key $PK_b$, resulting in the ciphertext $C_1$. Alice then conducts the additive homomorphic encryption by adding $C_1$ with $C_2$. Alice broadcasts the $C$ to Bob. 
        \item Alice can not know SA's private key, as nobody can guess the private key from the public key $PK_z$. It means Alice only knows where to send SA transactions but will never be able to log in to this SA wallet.
    \end{itemize}
    \item 
    \begin{itemize}
        \item  Bob can decrypt the ciphertext $C$ with his HE private key $sk_b$. The HE decryption result is the private key $sk_z$ to the wallet that receives the sent from Alice. 
        \item Then, he can generate the stealth address with $sk_z$ and decrypt it with the private key, which only Bob owns. So Bob can transfer its balance with the private key $sk_z$ for the SA wallet.

    \end{itemize}

\end{enumerate} 
\end{protocol}

\section{Implementation}
We have implemented our protocol in both HE and FHE scheme. We evaluated our implementations by comparing to DKSAP plain scheme. Both HE-DKSAP and FHE-DKSAP are variants of the DKSAP plain scheme. We can conclude the stealth addresses generation into three implementation processes: public and secret key generation, elliptic curve Diffie-Hellman (ECDH) algorithm, and stealth private key derivation as follows: 
%Our implementation of Homomorphic Encryption Stealth Address comprises two integral components. The first part encompasses the initial implementation of secp256k1, a widely recognized elliptic curve cryptography algorithm. This foundation forms the basis for secure cryptographic operations in our system. The second component involves the generation of stealth addresses, a crucial technique that enhances the privacy and anonymity of transactions. 

\begin{enumerate}
    \item Public and Secret Key Generation: We use the pseudorandom generator to generate the private keys in 256 bits. The public keys are generated based on the spec256k1 curve according to the private key to public key convert function. The public key with the Keccak-256 hash scheme then hashes the stealth address. 
    \item ECDH Algorithm: The Elliptic Curve Diffie-Hellman (ECDH) algorithm is pivotal in deriving shared secrets between parties. This cryptographic protocol targets the key exchange and ensures that secure communication channels are established while minimizing the risk of eavesdropping.
    \item Stealth Private Key Derivation: Building upon the generated keys and ECDH protocol, only the recipient can derive the private key of the stealth address, which preserver the privacy for this scheme.
\end{enumerate}
In this section, we will delve into the details of the HE-DKSAP with the Paillier scheme and the FHE-DKSAP with the BFV scheme.

\subsection{HE-DKSAP with Paillier}
Introducing HE using the Paillier scheme, our system extends the possibilities of secure computation without compromising data privacy. Integrating Paillier encryption requires a specific library, paillier, for seamless implementation. 
\begin{itemize}
    \item Paillier HE Key Generation: Bob's HE key is generated by key generation scheme. Based on the Section III primitive part, the public key pairs of Bob is as follows: 
    \begin{equation}
    \begin{aligned}
        PK_b = (n, g)\\
        sk_b = (\lambda, \mu)
    \end{aligned}
    \end{equation}
    \item Paillier Encryption and Homomorphic Addition: Choose a random $r$ where $r$ belongs to $\mathbb{Z}_n^*$. Taking the private key of Alice $sk_1$ and spending key of Bob $sk_2$, the ciphertext output as follows:
    \begin{equation}
    \begin{aligned}
        C_1 = g^{sk_1}r^n \bmod n^2 \\
        C_2 = g^{sk_2}r^n \bmod n^2
    \end{aligned}
    \end{equation}
    The addition of $C_1$ and $C_2$ meets the linear addition homomorphism property. 
    \begin{equation}
        C = C_1 + C_2 = g^{sk_1 + sk_2} r^n \bmod n^2 
    \end{equation}
    \item Paillier Decryption: Given a ciphertext $C$ and the private key $(\lambda, \mu)$, we can decrypt the ciphertext and get the result of $sk_z$ for stealth address generation. 
    \begin{equation}
         sk_z = sk_1 + sk_2=L\left(C ^\lambda \bmod n^2\right) \cdot \mu \bmod n
    \end{equation}
\end{itemize}

\subsection{FHE-DKSAP with BFV}
BFV is rooted in the LWE problem. In the FHE-DKSAP scheme, bootstrapping is not required. This is because the homomorphic addition is performed only once, and there is no risk of incorrect decryption due to excessive noise.

\begin{itemize}
    \item BFV Key Generation: Bob's HE key is generated by key generation scheme. Based on the Section III primitive part, we sample $\mathbf{sk_b} \leftarrow \chi$ and the public key of Bob is generated based on the LWE problem and shows as follows: 
    \begin{equation}
    \begin{aligned}
        PK_b = \left([-(\mathbf{a} \cdot \mathbf{sk_b}+\mathbf{e})]_q, \mathbf{a}\right), \mathbf{a} \leftarrow R_q, \mathbf{e} \leftarrow \chi
    \end{aligned}
    \end{equation}

\item BFV Encryption and Homomorphic Addition: Taking the private key of Alice $sk_1$ and spending key of Bob $sk_2$, the ciphertext output as follows:
    \begin{equation}
    \begin{aligned}
        \mathbf{C_1}=\left(\mathbf{pk}[0] \cdot u_1+e_1, \mathbf{pk}[1] \cdot u_1+e_1+\Delta \cdot sk_1\right) \\
        \mathbf{C_2}=\left(\mathbf{pk}[0] \cdot u_2+e_2, \mathbf{pk}[1] \cdot u_2+e_2+\Delta \cdot sk_2\right) 
    \end{aligned}
    \end{equation}
    The addition of $\mathbf{C_1}$ and $\mathbf{C_2}$ meets the addition homomorphism property and results in $\mathbf{C}$. 
    \begin{equation}
      \begin{aligned} 
      \mathbf{C} = \mathbf{C_1} + \mathbf{C_2} = \mathbf{C[0]} + \mathbf{C[1]} \\
    \mathbf{C}[0] = [\mathrm{pk}_0 \cdot(u_1+u_2)+(e_1+e_2)+\Delta(sk_1+sk_2)]_q.\\
    \mathbf{C}[1] =[\mathrm{pk}_1 \cdot(u_1+u_2)+(e_1+e_2)]_q)
    \end{aligned}
    \end{equation}
    \item BFV Decryption: Given a ciphertext $\mathbf{C}$ and the private key $(\mathbf{sk_b}$, we can decrypt the ciphertext and get the result of $sk_z$ for stealth address generation. 
    \begin{equation}
        sk_z^{\prime} \leftarrow \mathbf{C}[0] \cdot \mathbf{sk_b}+\mathbf{C}[1] \bmod q
    \end{equation}

\end{itemize}

\section{Evaluation}
In this section, we will prove the correctness of the FHE-DKSAP scheme and verify the security properties of HE-DKSAP in terms of data confidentiality, unlinkability, quantum computing attack resistance, and CPA security. 

\subsection{Correctness}
Correctness evaluation is essential for FHE-DKSAP due to introducing the learning with error assumption. We will demonstrate the correctness of FHE-DKSAP using the BFV scheme as an illustrative example, given that different encryption algorithms produce distinct ciphertexts. However, our FHE-DKSAP is capable of accurately computing any FHE schemes, including BFV, BGV \cite{brakerski2014leveled}, and the approximate FHE algorithm CKKS \cite{cheon2017homomorphic}.

\begin{proof}(Correctness)
    According to the encryption and decryption function of BFV scheme, for $\mathbf{sk_b} \cdot \mathbf{C}[0]+\mathbf{C}[1]=-e u+e_0 s+e_1+\Delta \cdot (sk_1 + sk_2)$. Let $e^{\prime}=-e u+e_0 s+e_1, \varepsilon=\frac{q}{t}-\Delta<1$, when $\left\|e^{\prime}\right\|_{\infty}<\frac{q}{2 t}-t$, we have:
$$
\left\lfloor\frac{t \cdot (sk_1+sk_2)^{\prime}}{q}\right\rceil=\left\lfloor\frac{t}{q}\left(\left(\frac{q}{t}-\varepsilon\right) (sk_1+sk_2)+e^{\prime}\right)\right\rceil  
$$
$$
\approx\left\lfloor (sk_1+sk_2)+\frac{t}{q}\left(e^{\prime}-(sk_1+sk_2)\right)\right\rceil=(sk_1+sk_2)
$$
It completes the correctness proof. 
\end{proof}

\subsection{Security analysis}

We have defined the security requirements of HE-DKSAP in Section IV. This section will prove that HE-DKSAP meets the requirements of confidentiality, unlinkability, post-quantum attack resistance, and CPA security. In HE-DKSAP, confidential data encompasses the secret keys of both Alice and Bob. Furthermore, the generated stealth address must remain unlinkable to the recipient's original address. This scheme relies on correctly executing encryption and decryption functions based on FHE, leveraging the inherent strengths of FHE to guard against quantum computing attacks. CPA security means the attacker can not distinguish the plaintext from the ciphertext. In this section, we will furnish formal proof for the security definitions associated with these requirements.

\subsubsection{Data confidentiality}
The security of a public key system lies in the fundamental principle that, given a public key, one cannot feasibly deduce the corresponding secret key. The public key is generated from the secret key based on secp256k1. The security of secp256k1 lies in the elliptic curve discrete logarithm problem (ECDLP). 
\begin{theorem}(Elliptic Curve Discrete Logarithm Problem (ECDLP))
    Let $E$ be an elliptic curve over a finite field $K$. Suppose there are points $P, Q \in E(K)$ given such that $Q \in\langle P\rangle$. Determine $k$ such that $Q=[k] P$.
\end{theorem}
\begin{proof}(Data confidentiality)
This ECDLP is regarded as an NP-hard problem, and many articles have proved this \cite{lauter2008elliptic, galbraith2016recent}. In our scheme, ECDLP preserves the data confidentiality of the secret keys. This completes the proof. 
\end{proof}

\subsubsection{Unlinkability}
Stealth addresses are typically one-time-use addresses. Whenever someone sends funds to a recipient, they compute and send it to a new stealth address. In our protocol, the SA is produced using large prime numbers, which are, in turn, generated by pseudorandom generator functions with distinct seeds, ensuring their uniqueness in each instance. This ensures that there is no common address on the blockchain to link multiple transactions to the same recipient. Since a new stealth address is created for every transaction, and it is not linked directly to the recipient's main public key on the blockchain, it becomes tough for an outsider to determine which transactions belong to a specific individual.

\subsubsection{Quantum computing attacks resistance}
Only FHE-DKSAP is quantum computing attack resistant due to the lattice-based cryptography \cite{barker2020getting}. The robustness of these cryptographic algorithms hinges on the computational difficulty of lattice problems, notably the Shortest Vector Problem (SVP) and the LWE problem. Crucially, this problem presents formidable challenges to both classical and quantum computing platforms.

\subsubsection{CPA security} Providing that the underlying homomorphic encryption scheme is CPA-secure, no attacker can find any bit of information on the data of the parties. In this section, we will prove that both HE-DKSAP and FHE-DKSAP are CPA-secure. 
\begin{itemize}
\item CPA security of HE-DKSAP (Paillier). Paillier is CPA-secure under the Decisional Composite Residuosity Assumption(DCRA) \cite{guo2020generalization}, which defines as follows:
\begin{definition}(Decisional Composite Residuosity Assumption(DCRA))
Given a composite $n$ and an integer $z$, it is hard to decide whether $z$ is an $n$-residue modulo $n^2$, i.e., whether there exists a $y$ such that
$$
z \equiv y^n \quad\left(\bmod n^2\right).
$$  
\end{definition}

\begin{proof}(CPA security of HE-DKSAP-Paillier) 
    Let $\mathcal{A}$ be a PPT CPA adversary for the Paillier's scheme. $D$ takes the role of the challenger of IND-CPA game for our scheme. $\mathcal{A}$ receives two ciphertext $C_1 = g^{sk_1}r^n \bmod n^2$ and $C_2 = g^{sk_2}r^n \bmod n^2$. Based on the DCRA, the adversary cannot distinguish $sk_1$ and $sk_2$ from $C_1$ and $C_2$. This completes our proof.
\end{proof}

    \item CPA security of FHE-DKSAP: FHE-DKSAP is based on the LEW assumption $\operatorname{LWE}\left(n_{lwe}, s, q\right)$, where $n_{lwe}, s, q$ is the parameter for security, consider matrix $A \stackrel{\$}{\leftarrow} \mathbb{Z}_q^{m \times n_{lwe}}$, vectors $r \stackrel{\$}{\leftarrow} \mathbb{Z}_q^{m \times 1}, x \stackrel{\mathrm{g}}{\leftarrow} \mathbb{Z}_{(0 s)}^{n_{lwe} \times 1}, e \stackrel{\mathrm{g}}{\leftarrow} \mathbb{Z}_{(0, s)}^{m \times 1}$. The advantage of a PPT adversary $\mathcal{A}$ in the challenge algorithm $\mathbf{D}$ can be converted as follows:
    \begin{equation}
    \begin{aligned}
    \mathbf{Adv}_{\mathcal{D}}^{\mathrm{LWE}\left(n_{l w e}, s, q\right)}(\lambda)= \\
    |\operatorname{Pr}[\mathcal{D}(A, A x+e) \rightarrow 1]-\operatorname{Pr}[\mathcal{D}(A, r) \rightarrow 1]|.
    \end{aligned}
    \end{equation}
\begin{proof}(CPA security of FHE-DKSAP) 
We start by providing adversary $\mathcal{A}$ with a random vector $P =pR-AS \stackrel{\$}{\leftarrow} \mathbb{Z}_q^{n_{lwe} \times l}$. To prove this $P$ is indistinguishable to adversary $\mathcal{A}$, we need to reduce $P$ into the LWE form, where 
\begin{equation}
    p^{-1} P=R+\left(-p^{-1} A\right) S \in \mathbb{Z}_q^{n_{l w e} \times l}.
\end{equation}
As $A$ is random, $A^{\prime}=-p^{-1} A \in \mathbb{Z}_q^{n_{l w e} \times n_{lwe}}$ is also random. Therefore, $P^{\prime}=p^{-1} P \in \mathbb{Z}_q^{n_{l w e} \times l}$ is random under the LWE assumption which in turn means $P$ is random as claimed.
    
\end{proof}
\end{itemize}
\section{Experimental validation}

\begin{table*}[]
\centering
\begin{tabular}{cl|l|l}
\multicolumn{1}{l}{}                                                                               & Information on chain   & Bits & Example                                                                                                                                                                                              \\ \hline
\multicolumn{1}{c|}{\multirow{3}{*}{DKSAP}}                                                        & PK\_scan               & 160  & \multirow{3}{*}{\begin{tabular}[c]{@{}l@{}}(0x86b1aa5120f079594348c67647679e7ac4c365b2c01330db782b0ba611c1d677, \\ 0x5f4376a23eed633657a90f385ba21068ed7e29859a7fab09e953cc5b3e89beba)\end{tabular}} \\ \cline{2-3}
\multicolumn{1}{c|}{}                                                                              & PK\_spend              & 160  &                                                                                                                                                                                                      \\ \cline{2-3}
\multicolumn{1}{c|}{}                                                                              & R\_a (public of Alice) & 160  &                                                                                                                                                                                                      \\ \hline
\multicolumn{1}{c|}{\multirow{3}{*}{\begin{tabular}[c]{@{}c@{}}FHE-DKSAP\\ Paillier\end{tabular}}} & PK\_bob                & 160  & \begin{tabular}[c]{@{}l@{}}(0x86b1aa5120f079594348c67647679e7ac4c365b2c01330db782b0ba611c1d677, \\ 0x5f4376a23eed633657a90f385ba21068ed7e29859a7fab09e953cc5b3e89beba)\end{tabular}                  \\ \cline{2-4} 
\multicolumn{1}{c|}{}                                                                              & PK\_fhe\_bob           & 128  & 192ace432e                                                                                                                                                                                           \\ \cline{2-4} 
\multicolumn{1}{c|}{}                                                                              & $C_2, C$                     & 48   & 0x7faf7cf217c0                                                                                                                                                                                       \\ \hline
\multicolumn{1}{c|}{\multirow{3}{*}{\begin{tabular}[c]{@{}c@{}}FHE-DKSAP\\ BFV\end{tabular}}}      & PK\_bob                & 160  & \begin{tabular}[c]{@{}l@{}}(0x86b1aa5120f079594348c67647679e7ac4c365b2c01330db782b0ba611c1d677, \\ 0x5f4376a23eed633657a90f385ba21068ed7e29859a7fab09e953cc5b3e89beba)\end{tabular}                  \\ \cline{2-4} 
\multicolumn{1}{c|}{}                                                                              & PK\_fhe\_bob           & 128  & hi+SltsSwYnILVvNl5mFp+jbKJnlxwg7r7g1DGr8QQs=                                                                                                                                                         \\ \cline{2-4} 
\multicolumn{1}{c|}{}                                                                              & $C_2, C$                     & 256  & Na*                                                                                                                                                                                                  \\ \hline
\end{tabular}
\caption{Information storage on Chain}
\label{tab:storage}
\end{table*}

To thoroughly assess the performance and effectiveness of our system, we conducted comprehensive test of the generated stealth addresses in terms of the SAP computation time and storage of the generated stealth addresses. Specifically, we evaluated the stealth address generation process using three different setups: DKSAP (Plain), HE-DKSAP (Pallier), and FHE-DKSAP (BFV). 

The following settings are used in the experiments. Processor: Linux, 2.3 GHz Quad-Core Intel Core i5; memory: 8 GB 2133 MHz LPDDR3; python Version 3.9 with the libraries of Python-Paillier 1.2.2 and Concrete, zamafhe/concrete-python:v2.0.0.

\subsection{Computation time}

The computation time benchmark can be found in Table \ref{tab:computation time}. We run each of the algorithms 100 times. We can observe that DKSAP excels in computational speed due to its lack of privacy-preserving encryption. HE-DKSAP-Paillier balances enhanced data privacy with longer computational time due to the intricate encryption and decryption of the Paillier scheme, which is around 20 times slower than the plain scheme. FHE-DKSAP-BFV is slightly slower than unencrypted DKSAP but notably faster than HE-DKSAP-Paillier. This efficiency can be attributed to its implementation in the RUST programming language, highlighting the importance of suitable tools for execution.
\begin{table}[htb]
    \centering
\begin{tabular}{lllll} 
& DKSAP & HE-DKSAP & FHE-DKSAP \\
& Plain & Pallier & BFV \\
\hline Average (s) & 0.01938 & 0.44560 & 0.0355 \\
\hline Max (s) & 0.02218 & 0.9829 & 0.05095 \\
\hline Min (s) & 0.01751 & 0.10880 & 0.02870 \\
\hline
\end{tabular}
    \caption{Computation time benchmark}
    \label{tab:computation time}
\end{table}

\subsection{On chain storage}

The on-chain storage of DKSAP-Plain, HE-DKSAP-Paillier, and FHE-DKSAP-BFV can be found in Table \ref{tab:storage}. Several notable results can be gleaned from the provided tables. Notably, a comparison between HE-DKSAP-Paillier and FHE-DKSAP-BFV against DK-SAP-Plain reveals that HE-DKSAP-Paillier and FHE-DKSAP-BFV require less storage space. This variance in storage consumption can be attributed to the use of distinct encryption and decryption algorithms. It is worth noting that the storage for $C1$ in FHE-DKSAP-BFV is a two-part composed array. 

\section{Conclusion}
Inspired by the concepts of DKSAP and BaseSAP, we introduce the HE-DKSAP protocol. This can be further extended to FHE-DKSAP, designed to prevent associating a blockchain transaction's output with a recipient's public address. Through our security analysis, we demonstrate that this protocol upholds the security properties of data confidentiality and unlinkability, and is resistant to quantum computing and CPA security. Performance test results reveal that it requires the minimum on-chain storage, although HE-DKSAP introduces a modest increase in computation time.

However, as we look towards our next phase of research, an interesting avenue to explore is a comparative analysis among various FHE schemes to understand the influence and implications of this protocol more deeply. Another pertinent question is the effective management of key pairs within the wallet. Proper key management is pivotal not only for operational efficiency but also for ensuring the security and integrity of transactions. Evaluating different strategies for key storage, rotation and retrieval can enhance the robustness of our protocol and provide valuable insights into its practical applications.

\printcredits

\bibliographystyle{elsarticle-num}
% \bibliographystyle{cas-model2-names}
% \bibliographystyle{cas-model1-names}
%\bibliographystyle{model-num}

\bibliography{cas-refs}


\end{document}